\documentclass[%
reprint,
nofootinbib,
 amsmath,amssymb,
 aps,
 pra,
showkeys,
]{revtex4-2}

\usepackage{mathrsfs}

\usepackage{graphicx}
\usepackage{dcolumn}
\usepackage{bm}

\usepackage{amsthm}
\usepackage{xfrac}
\usepackage[mathcal]{euscript}
\usepackage{tikz}
\usetikzlibrary{trees,arrows,decorations.pathreplacing}

\usepackage{url}
\usepackage{hyperref}
\usepackage{color}
\definecolor{refcolor}{RGB}{0,0,190}
\hypersetup{
    colorlinks,
    citecolor=refcolor,
    filecolor=refcolor,
    linkcolor=refcolor,
    urlcolor=refcolor
}

\usepackage{booktabs}

\usepackage{bookmark}
\bookmarksetup{
  numbered, 
  open,
}

\newtheorem{theorem}{Theorem}

\newtheorem{corollary}{Corollary}

\theoremstyle{remark}
\newtheorem{remark}{Remark}
\newtheorem{example}{Example}

\theoremstyle{definition}
\newtheorem{observation}{Observation}

\newtheorem{objection}{Objection}
\newtheorem{reply}{Reply}

\theoremstyle{definition}
\newtheorem{defCustom}{Definition}
\renewcommand{\thedefCustom}{\arabic{definition}}
\makeatletter
\newcommand{\setdefCustomtag}[1]{
  \let\oldthedefCustom\thedefCustom
  \renewcommand{\thedefCustom}{#1}
  \g@addto@macro\enddefCustom{
    \global\let\thedefCustom\oldthedefCustom}
  }
\makeatother

\newtheorem{assumption}{Assumption}
\renewcommand{\theassumption}{\arabic{assumption}}
\makeatletter
\newcommand{\setassumptiontag}[1]{
  \let\oldtheassumption\theassumption
  \renewcommand{\theassumption}{#1}
  \g@addto@macro\endassumption{
    \global\let\theassumption\oldtheassumption}
  }
\makeatother

\newtheorem{claim}{Claim}
\renewcommand{\theclaim}{\arabic{claim}}
\makeatletter
\newcommand{\setclaimtag}[1]{
  \let\oldtheclaim\theclaim
  \renewcommand{\theclaim}{#1}
  \g@addto@macro\endclaim{
    \global\let\theclaim\oldtheclaim}
  }
\makeatother

\theoremstyle{remark}
\newtheorem{pointItem}{}
\renewcommand{\thepointItem}{\quad\arabic{pointItem}}
\makeatletter
\newcommand{\setpointItemtag}[1]{
  \let\oldthepointItem\thepointItem
  \renewcommand{\thepointItem}{#1}
  \g@addto@macro\endpointItem{
    \global\let\thepointItem\oldthepointItem}
  }
\makeatother

\begin{document}


\newcommand{\pbref}[1]{\ref{#1} (\nameref*{#1})}
   
\def\({\big(}
\def\){\big)}

\newcommand{\tn}{\textnormal}
\newcommand{\ds}{\displaystyle}
\newcommand{\dsfrac}[2]{\displaystyle{\frac{#1}{#2}}}

\newcommand{\boplus}{\textstyle{\bigoplus}}
\newcommand{\botimes}{\textstyle{\bigotimes}}
\newcommand{\bcup}{\textstyle{\bigcup}}
\newcommand{\bsqcup}{\textstyle{\bigsqcup}}
\newcommand{\bcap}{\textstyle{\bigcap}}

\newcommand{\MQS}{\ref{def:MQS}}

\newcommand{\struct}{\mc{S}}
\newcommand{\kind}{\mc{K}}

\newcommand{\dddots}{\rotatebox[origin=t]{135}{$\cdots$}}

\newcommand{\statespace}{\mathcal{S}}
\newcommand{\hilbert}{\mathcal{H}}
\newcommand{\vectorspace}{\mathcal{V}}
\newcommand{\mc}[1]{\mathcal{#1}}
\newcommand{\ms}[1]{\mathscr{#1}}
\newcommand{\mf}[1]{\mathfrak{#1}}
\newcommand{\dU}{\wh{\mc{U}}}

\newcommand{\wh}[1]{\widehat{#1}}
\newcommand{\dwh}[1]{\wh{\rule{0ex}{1.3ex}\smash{\wh{\hfill{#1}\,}}}}

\newcommand{\wt}[1]{\widetilde{#1}}
\newcommand{\wht}[1]{\widehat{\widetilde{#1}}}
\newcommand{\on}[1]{\operatorname{#1}}

\newcommand{\qmU}{$\mathscr{U}$}
\newcommand{\qmR}{$\mathscr{R}$}
\newcommand{\qmUR}{$\mathscr{UR}$}
\newcommand{\qmDR}{$\mathscr{DR}$}

\newcommand{\R}{\mathbb{R}}
\newcommand{\C}{\mathbb{C}}
\newcommand{\Z}{\mathbb{Z}}
\newcommand{\K}{\mathbb{K}}
\newcommand{\N}{\mathbb{N}}
\newcommand{\Prj}{\mathcal{P}}
\newcommand{\abs}[1]{|#1|}

\newcommand{\de}{\operatorname{d}}
\newcommand{\tr}{\operatorname{tr}}
\newcommand{\im}{\operatorname{Im}}

\newcommand{\ie}{\textit{i.e.}\ }
\newcommand{\vs}{\textit{vs.}\ }
\newcommand{\eg}{\textit{e.g.}\ }
\newcommand{\cf}{\textit{cf.}\ }
\newcommand{\etc}{\textit{etc}}
\newcommand{\etal}{\textit{et al.}}

\newcommand{\Span}{\tn{span}}
\newcommand{\pde}{PDE}
\newcommand{\U}{\tn{U}}
\newcommand{\SU}{\tn{SU}}
\newcommand{\GL}{\tn{GL}}

\newcommand{\schrod}{Schr\"odinger}
\newcommand{\vonneum}{Liouville-von Neumann}
\newcommand{\ks}{Kochen-Specker}
\newcommand{\leggarg}{Leggett-Garg}
\newcommand{\bra}[1]{\langle#1|}
\newcommand{\ket}[1]{|#1\rangle}
\newcommand{\kett}[1]{|\!\!|#1\rangle\!\!\rangle}
\newcommand{\proj}[1]{\ket{#1}\bra{#1}}
\newcommand{\braket}[2]{\langle#1|#2\rangle}
\newcommand{\ketbra}[2]{|#1\rangle\langle#2|}
\newcommand{\expectation}[1]{\langle#1\rangle}
\newcommand{\Herm}{\tn{Herm}}
\newcommand{\Sym}[1]{\tn{Sym}_{#1}}
\newcommand{\meanvalue}[2]{\langle{#1}\rangle_{#2}}
\newcommand{\Prob}{\tn{Prob}}
\newcommand{\kjj}[3]{#1\!:\!#2,#3}
\newcommand{\jk}[2]{#1,#2}
\newcommand{\JK}{\mf{j}}

\newcommand{\weightU}[5]{\big[{#2}{}_{#3}\overset{#1}{\rightarrow}{#4}{}_{#5}\big]}
\newcommand{\weightUT}[8]{\big[{#3}{}_{#4}\overset{#1}{\rightarrow}{#5}{}_{#6}\overset{#2}{\rightarrow}{#7}{}_{#8}\big]}
\newcommand{\weight}[4]{\weightU{}{#1}{#2}{#3}{#4}}
\newcommand{\weightT}[6]{\weightUT{}{}{#1}{#2}{#3}{#4}{#5}{#6}}

\newcommand{\btimes}{\boxtimes}
\newcommand{\btimess}{{\boxtimes_s}}

\newcommand{\h}{\mathbf{(2\pi\hbar)}}
\newcommand{\x}{\mathbf{x}}
\newcommand{\xThree}{\boldsymbol{x}}
\newcommand{\z}{\mathbf{z}}
\newcommand{\q}{\mathbf{q}}
\newcommand{\p}{\mathbf{p}}
\newcommand{\0}{\mathbf{0}}
\newcommand{\annih}{\widehat{\mathbf{a}}}

\newcommand{\cs}{\mathscr{C}}
\newcommand{\ps}{\mathscr{P}}
\newcommand{\xhat}{\widehat{\x}}
\newcommand{\phat}{\widehat{\mathbf{p}}}
\newcommand{\fqproj}[1]{\Pi_{#1}}
\newcommand{\cqproj}[1]{\wh{\Pi}_{#1}}
\newcommand{\cproj}[1]{\wh{\Pi}^{\perp}_{#1}}

\newcommand{\M}{\mathbb{E}_3}
\newcommand{\D}{\mathbf{D}}
\newcommand{\dn}{\tn{d}}
\newcommand{\db}{\mathbf{d}}
\newcommand{\n}{\mathbf{n}}
\newcommand{\m}{\mathbf{m}}
\newcommand{\V}[1]{\mathbb{V}_{#1}}
\newcommand{\F}[1]{\mathcal{F}_{#1}}
\newcommand{\Fvacuumfield}{\widetilde{\mathcal{F}}^0}
\newcommand{\nD}[1]{|{#1}|}
\newcommand{\Lin}{\mathcal{L}}
\newcommand{\End}{\tn{End}}
\newcommand{\vbundle}[4]{{#1}\to {#2} \stackrel{\pi_{#3}}{\to} {#4}}
\newcommand{\vbundlex}[1]{\vbundle{V_{#1}}{E_{#1}}{#1}{M_{#1}}}
\newcommand{\rep}{\rho_{\scriptscriptstyle\btimes}}

\newcommand{\intl}[1]{\int\limits_{#1}}

\newcommand{\moyalBracket}[1]{\{\mskip-5mu\{#1\}\mskip-5mu\}}

\newcommand{\Hint}{H_{\tn{int}}}

\newcommand{\quot}[1]{``#1''}

\def\sref #1{\S\ref{#1}}

\newcommand{\dBB}{de Broglie--Bohm}
\newcommand{\dBBt}{{\dBB} theory}
\newcommand{\pwt}{pilot-wave theory}
\newcommand{\PWT}{PWT}
\newcommand{\NRQM}{{\textbf{NRQM}}}

\newcommand{\image}[3]{
\begin{center}
\begin{figure}[!ht]
\includegraphics[width=#2\textwidth]{#1}
\caption{\small{\label{#1}#3}}
\end{figure}
\end{center}
\vspace{-0.40in}
}

\newcommand{\todo}[1]{\textcolor{red}{$\Rightarrow$} \textcolor{blue}{#1}\PackageWarning{TODO:}{#1!}}

\title{The Problem of Irreversible Change in Quantum Mechanics}


\author{Ovidiu Cristinel Stoica}
\affiliation{
 Dept. of Theoretical Physics, NIPNE---HH, Bucharest, Romania. \\
	Email: \href{mailto:cristi.stoica@theory.nipne.ro}{cristi.stoica@theory.nipne.ro},  \href{mailto:holotronix@gmail.com}{holotronix@gmail.com}
	}%

\date{\today}

\begin{abstract}
I prove that, if a change happens in a closed quantum system so that its state is perfectly distinguishable from all past or future states, the Hamiltonian is $\widehat{H}=-i\hbar\frac{\partial\ }{\partial\tau}$. A time operator $\widehat{\tau}$ can be defined as its canonical conjugate. This Hamiltonian is usually rejected because it has no ground state, but I show that even a weaker form of irreversibility is inconsistent with a ground state.


What is the right choice, that the world's Hamiltonian is $-i\hbar\frac{\partial\ }{\partial\tau}$, or that changes are reversible?
\end{abstract}


\maketitle

\section{Introduction}
\label{s:intro}

Our universe is in a permanent change, at any time its state is different from any past or future state.
It continuously expands, its Boltzmann entropy increases, matter rearranges itself in new patterns.
Every quantum measurement leads to irreversible changes as well.

In this article, I look into the constraints imposed by irreversibility to the physical law itself.
In Sec. \sref{s:change} I prove that even a single irreversible change requires the Hamiltonian to have the form
\begin{equation}
\label{eq:time_translation_Hamiltonian}
\widehat{H}=-i\hbar\frac{\partial\ }{\partial\tau}.
\end{equation}

This leads to a dilemma: we expect the Hamiltonian of our world to be more complex and to have a ground state, we also expect that irreversible changes happen all the time, but we cannot have both.

I show how several attempts to avoid the dilemma fail:

In Sec. \sref{s:weak} I show that the Hamiltonian cannot have a ground state even if the irreversible change is allowed to take some time to happen, and I argue that using POVM cannot avoid the problem.

In Sec. \sref{s:no_exit} I show that, although the collapse postulate or decoherence can be used to escape into a larger subspace, this does not avoid the dilemma.

Some arguments against the Hamiltonian \eqref{eq:time_translation_Hamiltonian} are discussed in Sec. \sref{s:objections}, based on examples of quantum systems proved in \cite{Stoica2022VersatilityOfTranslations} to have the Hamiltonian of the form \eqref{eq:time_translation_Hamiltonian}.
These examples include
the Hamiltonian of the standard model of ideal quantum measurements,
any quantum world in which there is a perfect clock or a sterile massless fermion in a certain state,
and the quantum representations of all deterministic time-reversible dynamical systems without time loops.

Section \sref{s:conclusions} summarizes the conclusions.

\section{Irreversible Change Theorem}
\label{s:change}

Consider a closed quantum system, represented by a state vector $\ket{\psi}$ in a Hilbert space $\hilbert$, which evolves according to the {\schrod} equation
\begin{equation}
\label{eq:schrod}
i\hbar\frac{d\ }{d t}\ket{\psi(t)}=\wh{H}\ket{\psi(t)}.
\end{equation}

The Hamiltonian operator $\wh{H}$ is time-independent.
For the initial condition $\ket{\psi(0)}=\ket{\psi_0}$,
\begin{equation}
	\label{eq:unitary_evolution}
	\ket{\psi(t)}=\wh{U}(t)\ket{\psi_0}, \tn{ where }
	\wh{U}(t):=e^{-\frac{i}{\hbar}t\wh{H}}.
\end{equation}

Quantum states are not directly observable. All we can observe are properties of the macro states. For example, quantum measurements lead to macroscopic changes of the pointer state, which can then be observed.
Macro states are classes of equivalence of micro states, represented by subspaces $\hilbert_k$ in the Hilbert space $\hilbert$, or, equivalently, by projectors $\wh{P}_k$ determining them, $\hilbert_k=\wh{P}_k\hilbert$.
The system evolves through a succession of macro states.

Irreversibility of change requires that, once the system leaves a macro state, it never revisits it again. This means that all its macro states at distinct times should be orthogonal.
If strict orthogonality fails, the state vector has at least a component whose macro state is a past or a future macro state, therefore there is a nonzero probability of spontaneous time travel in the past or in the future.
Moreover, even when the state vector evolves unitarily into a superposition of macroscopically distinct states, as it happens due to measurements, irreversibility requires all these macro states to be distinct from the past ones.
We are led to the following Observation:
\begin{observation}
\label{obs:ortho}
Once a system leaves a macro state, if its state is merely different, and not strictly orthogonal to its past states, it can return to a previous macro state.
\end{observation}

For more generality, I will assume orthogonality of micro rather than macro states.
This condition is weaker, because orthogonality of macro states implies orthogonality of micro states, but the converse is not true.

\begin{theorem}
\label{thm:translation}
If the state of a quantum system at $t=0$ is perfectly distinguishable from all its past (or future) states, then the states in its history span a Hilbert space on which the restriction of its Hamiltonian is $-i\hbar\frac{\partial\ }{\partial\tau}$.
\end{theorem}
\begin{proof}
Assume that $\ket{\psi(t)}\perp \ket{\psi(0)}$ for all $t<0$ (the case $t>0$ is similar).
Let $t_1 < t_2\in\R$.
Then, $t_1-t_2<0$, and since $\ket{\psi(t_1-t_2)}\perp\ket{\psi(0)}$,
\begin{equation}
\label{eq:past_future_spaces_perp}
\begin{aligned}
\braket{\psi(t_2)}{\psi(t_1)}
&=\bra{\psi(0)}\wh{U}(t_2)^\dagger\wh{U}(t_1)\ket{\psi(0)}\\
&=\bra{\psi(0)}\wh{U}(t_1-t_2)\ket{\psi(0)}\\
&=0.
\end{aligned}
\end{equation}

Therefore, $\(\ket{\psi(t)}\)_{t\in\R}$ is a complete basis of the vector space $\mc{V}:=\Span\{\ket{\psi(t)}|t\in\R\}\cong L^2(\R,\C)$, which is thus an invariant subspace for $\(\wh{U}(t)\)_{t\in\R}$.
Since $\mc{V}$ is a closed subspace of $\hilbert$, it is a Hilbert space. 
The one-parameter group $\(\wh{U}(t)\)_{t\in\R}$ acts on $\mc{V}$ like a translation group. From Stone's Theorem \cite{Stone1932OnOneParameterUnitaryGroupsInHilbertSpace}, the unique Hamiltonian operator generating its restriction to $\mc{V}$ is $-i\hbar\frac{\partial\ }{\partial\tau}$.
\end{proof}

Theorem \ref{thm:translation} shows that if even a single irreversible change occurs, on the Hilbert space spanned by the history of the state the Hamiltonian is $-i\hbar\frac{\partial\ }{\partial\tau}$.

\begin{corollary}
\label{thm:time_operator}
With the assumptions of Theorem \ref{thm:translation}, the Hamiltonian admits a canonical conjugate $\wh{\tau}$.
For each $\tau\in\R$, $\ket{\psi(\tau)}$ is an eigenstate corresponding to the eigenvalue $\tau$ of $\wh{\tau}$, so $\wh{\tau}$ is a \emph{time operator}.
\end{corollary}
\begin{proof}
Since  $-i\hbar\frac{\partial\ }{\partial\tau}$ generates translations between orthogonal state vectors, the Stone-von Neumann Theorem implies that it admits a canonical conjugate, whose representation is $\bra{\tau}\wh{\tau}\ket{\psi} = \tau\braket{\tau}{\psi}$, $\tau\in\R$  \cite{Weyl1927QuantenmechanikUndGruppentheorie}. So $\ket{\psi(\tau)}$ are the eigenstates of $\wh{\tau}$, which makes it a time operator.
\end{proof}

In the following, I will discuss whether and at what costs it is possible to avoid the conclusion of Theorem \ref{thm:translation}.

\section{Do weaker assumptions help?}
\label{s:weak}

Theorem \ref{thm:translation} implies that the Hamiltonian has no ground state.
Can this conclusion be avoided by weakening the assumption of Theorem \ref{thm:translation}, that the state is perfectly distinguishable from its past or future states? I will look into two possible ways to do this.

One way is to allow the irreversible changes take some time to happen.
The minimum time for a quantum system to evolve between two distinguishable states is called \emph{quantum speed limit} (QSL)
\cite{MargolusLevitin1998MaximumSpeedOfDynamicalEvolution,DeffnerCampbell2017QuantumSpeedLimits}. For a system with the Hamiltonian from eq. \eqref{eq:time_translation_Hamiltonian} and whose state is an eigenstate of $\wh{\tau}$, the QSL is zero, the mean energy is undetermined and its variance is infinite.
From the most cases studied in the literature we may think that the QSL has to be positive, but this is of course because they have different Hamiltonians than \eqref{eq:time_translation_Hamiltonian}.

Now we will see that even if the QSL is strictly positive, irreversibility of change still implies that the system cannot have a ground state.

Unruh and Wald \cite{UnruhWald1989TimeAndTheInterpretationOfCanonicalQuantumGravity} showed that a weaker time operator is not possible for a system having a ground state, \ie there is no self-adjoint operator $\wh{T}$ so that the system is in distinct eigenstates of $\wh{T}$ at discrete times $t_0<t_1<t_2\ldots$ with zero chance to return to a previous eigenstate.

But in fact their proof shows much more than stated.

\begin{theorem}
\label{thm:discrete_change}
Consider a quantum system evolving from the state $\ket{s_1}$ at the time $t_1$ to $\ket{s_2}$ at $t_2>t_1$, possibly suffering a projection, \ie $\bra{s_2}\wh{U}(t_2-t_1)\ket{s_1}\neq 0$.
If $\ket{s_1}$ and $\ket{s_2}$ are perfectly distinguishable and the probability that the system returns to the state $\ket{s_1}$ after $t_2$ is $0$, then it cannot have a ground state.
\end{theorem}
\begin{proof}
It is the same proof as for Unruh \& Wald's Theorem \cite{UnruhWald1989TimeAndTheInterpretationOfCanonicalQuantumGravity}.
If $\wh{H}$ has a ground state, the function $f:\C\to\C$,
\begin{equation}
\label{eq:uw_function}
f(z):=\bra{s_1}e^{-\frac{i}{\hbar}z\wh{H}}\ket{s_2}
\end{equation}
is holomorphic in the lower-half complex plane \cite{StreaterWightman1964PCTSpinStatisticsAndAllThat}.
The condition that the system cannot return to $\ket{s_1}$ after $t_2$ is $f(t)=0$ for all real $t>0$.
This implies that $f(z)=0$ in the closed lower half-plane, where $f$ is holomorphic. Hence, $f(t)=0$ for any $t\in\R$.
Then, $\bra{s_2}\wh{U}(t)\ket{s_1}=\bra{s_1}\wh{U}(-t)\ket{s_2}^\ast=f^\ast(-t)=0$ for all $t>0$, contradicting the assumption that $\bra{s_2}\wh{U}(t_2-t_1)\ket{s_1}\neq 0$.
\end{proof}

Therefore, even if we weaken the condition of irreversibility of change from Theorem \ref{thm:translation} to a bare minimum, it still requires that there is no ground state.

Another way we may try to circumvent the dilemma is to use \emph{positive operator-valued measures} (POVM). POVM can discriminate nonorthogonal quantum states, but imperfectly, so we have to accept a nonnegligible probability to fail to distinguish them, and to  know what states to expect \cite{Ivanovic1987HowToDifferentiateBetweenNonorthogonalStates,Dieks1988OverlapAndDistinguishabilityOfQuantumStates}.
Since we can know when the measurement failed, we can try again on another copy of the system until we succeed. But this does not help when we monitor changes at the macro level of the world, because we cannot prepare multiple copies of the world itself.

An interesting result based on POVM is the theoretical construction of a clock whose Hamiltonian has a ground state, see \eg \cite{Holevo2011ProbabilisticAndStatisticalAspectsOfQuantumTheory} and more recently \cite{LoveridgeMiyadera2019RelativeQuantumTime,HohnSmithLock2021TrinityOfRelationalQuantumDynamics}.
The resulting POVM is translation-invariant.
Unfortunately, if the POVM is not projective, which is always the case when there is a ground state, the clock states overlap.

POVM clocks are very simple closed systems, but we can hope that the idea can be used to obtain more complex systems.
Omn{\`e}s \cite{Omnes1997QuantumClassicalCorrespondenceUsingProjectionOperators} did this to represent the macro states of the world as \emph{quasiprojection operators}.
But unlike the genuine projectors representing macro states, the quasiprojection operators overlap, preventing the irreversibility of change at the macro level.

Using POVM can be very useful, but the existence of overlapping implies that the system has a nonzero probability to return to a previous macro state.
And this means the violation of irreversibility of change.

\section{Is collapse an exit?}
\label{s:no_exit}

Theorem \ref{thm:translation} assumed unitary evolution, and the result applies to the subspace $\mc{V}:=\Span\{\ket{\psi(t)}|t\in\R\}$.
But if the wavefunction collapses or, depending on the interpretation, it branches due to decoherence, the state vector leaves the subspace $\mc{V}$.
Can collapse avoid the consequences of Theorem \ref{thm:translation} by reaching a larger subspace of the Hilbert space where the Hamiltonian is not \eqref{eq:time_translation_Hamiltonian}?

First, even if the Hamiltonian is different from $-i\hbar\frac{\partial\ }{\partial\tau}$ outside of $\mc{V}$, since its restriction to $\mc{V}$ does not have a ground state, the full Hamiltonian also does not have a ground state.
This already shows that the dilemma cannot be avoided even if the state vector exits $\mc{V}$.

Second, while wavefunction collapse or decoherence can allow the state vector to leave the subspace $\mc{V}$, Theorem \ref{thm:translation} applies as well to the branch that exited $\mc{V}$. Therefore, collapse or branching only extends the subspace where Theorem \ref{thm:translation} applies and the Hamiltonian is $-i\hbar\frac{\partial\ }{\partial\tau}$.

To see this, let us consider what happens during the measurement of a system $S$ with a measuring device $M$. Let the observable whose value we measure on $S$ have the possible eigenstates $\ket{\lambda}_S$ with eigenvalues $\lambda$. The measuring device $M$ is assumed to start in the ``ready'' state $\ket{\tn{ready}}_{M}$, and to end in one of the states $\ket{\tn{result=}\lambda}_{M}$ depending on the outcome $\lambda$, resulting in a superposition
\begin{equation}
\label{eq:measurement}
\ket{\psi}_{S}\ket{\tn{ready}}_{M} \stackrel{\wh{U}}{\mapsto} \sum_{\lambda}\braket{\lambda}{\psi}_{S}\ket{\lambda}_{S}\ket{\tn{result=}\lambda}_{M}.
\end{equation}

The observer should perceive only one possible outcome. This is usually achieved by assuming a projection, so that only one term in the superposition in the RHS of eq. \eqref{eq:measurement} remains. Alternatively, we can assume that the total wavefunction decohered into branches containing instances of the observer, each observing only the outcome from that branch, as in the Everett approach \cite{Everett1957RelativeStateFormulationOfQuantumMechanics}.
But, for any $\lambda$, $\braket{\tn{result=}\lambda}{\tn{ready}}_M=0$.
This implies that all terms in the RHS of eq. \eqref{eq:measurement} are orthogonal to the initial state, and so is their superposition,
\begin{equation}
\label{eq:measuring_ortho_tot}
\Big(\sum_{\lambda}\braket{\psi}{\lambda}_{S}\bra{\lambda}_{S}\bra{\tn{result=}\lambda}_{M}\Big)\ket{\psi}_{S}\ket{\tn{ready}}_{M}=0.
\end{equation}

Therefore, even after the measurement, Theorem \ref{thm:translation} continues to apply to the subspace $\mc{V}$ spanned by the orbit of the initial state vector under the action of the unitary evolution group. In addition, it also applies to the subspaces spanned by the orbits of the components of the original state vector resulting from projection (or branching), and to the space spanned by the orbits of all of these components, which includes $\mc{V}$.

Since new orbits result with each measurement, the applicability domain of Theorem \ref{thm:translation} extends to the Hilbert subspace spanned by all these orbits, and on this subspace, either the Hamiltonian is again $-i\hbar\frac{\partial\ }{\partial\tau}$, or irreversibility of change is violated.

\section{Objections against translational quantum systems}
\label{s:objections}

In this Section I briefly discuss some known objections against the Hamiltonian $\wh{H} = -i\hbar\frac{\partial\ }{\partial\tau}$, following \cite{Stoica2022VersatilityOfTranslations}.

\begin{objection}
\label{obj:simple}
The Hamiltonian $-i\hbar\frac{\partial\ }{\partial\tau}$ is too simple to represent physically realistic quantum systems.
\end{objection}
\begin{reply}
\label{reply:simple}
In \cite{Stoica2022VersatilityOfTranslations} I showed that there are numerous examples of quantum systems with this Hamiltonian:

\begin{example}
\label{tqs:measurement}
Any closed system consisting of a measuring device and an observed system realizing the \emph{standard model of ideal (projective) measurements}, as described in \cite{BuschGrabowskiLahti1995OperationalQuantumPhysics} \S II.3.4 and \cite{Mittelstaedt2004InterpretationOfQMAndMeasurementProcess} \S 2.2(b). This is limited to the time during which the measurement takes place, for systems that do not have free evolution. Although in a basis which is not consistent with the factorization corresponding to the subsystems, the Hamiltonian is $-i\hbar\frac{\partial\ }{\partial\tau}$.
\end{example}

\begin{example}
\label{tqs:clock}
Any world containing a closed subsystem with the Hamiltonian \eqref{eq:time_translation_Hamiltonian}. This can be for example an ideal clock or a sterile massless fermion in a certain state described in \cite{Stoica2022VersatilityOfTranslations}.
\end{example}

\begin{example}
\label{tqs:dyn_sys}
The Koopman \cite{Koopman1931HamiltonianSystemsAndTransformationInHilbertSpace} quantum representation of any deterministic time-reversible dynamical system without time loops.
In \cite{Stoica2022VersatilityOfTranslations} I also showed that all the properties of the original dynamical system and its evolution equations are faithfully encoded as quantum observables.
\end{example}

All these systems have the same simple Hamiltonian \eqref{eq:time_translation_Hamiltonian}.
But the correspondence between observables and physical properties can make them very different, ensuring unlimited diversity and complexity \cite{Stoica2022VersatilityOfTranslations}.
There seems to be no limit to how similar to the real world such a system can be.
\qed
\end{reply}

\begin{objection}
\label{obj:infinite}
Without a ground state, the system decays indefinitely towards negative energy states.
One can imagine using such systems to generate free energy.
\end{objection}
\begin{reply}
\label{reply:infinite}
Consider Example \ref{tqs:dyn_sys}.
If the original dynamical system does not decay indefinitely, its quantum representation, which represents faithfully the evolution of the original dynamical system, does not decay indefinitely either,
despite having the Hamiltonian \eqref{eq:time_translation_Hamiltonian}.
Also such systems cannot be used to extract infinite amounts of free energy, for example by charging infinitely many batteries.
If this were true, this would imply that we could do the same with deterministic time-reversible dynamical systems without time loops, but this is not possible.

From Example \ref{tqs:clock}, given a quantum system $R$ composed of a quantum system with a ground state and an ideal clock $C$ which do not interact, the total system has the Hamiltonian \eqref{eq:time_translation_Hamiltonian}. Yet, the system $R$ and the clock $C$ do not decay indefinitely and cannot be used to extract infinite amounts of free energy. A similar conclusion follows for the case of the massless sterile fermion.
\qed
\end{reply}

\begin{objection}
\label{obj:negative}
There are strong indications that in our world the Hamiltonian has a ground state, unlike \eqref{eq:time_translation_Hamiltonian}.
\end{objection}
\begin{reply}
\label{reply:negative}
There are such indications. On the other hand,

-- Objection \ref{obj:negative}, raised \eg by Pauli \cite{Pauli1980GeneralPrinciplesOfQuantumMechanics} against the existence of a self-adjoint time operator $\wh{\tau}$ which is the canonical conjugate of the Hamiltonian, was based on nonrelativistic Quantum Mechanics.
While quantizing the relativistic electron, Dirac found a Hamiltonian without a ground state. He introduced the \emph{Dirac sea} to prevent the decay into infinite negative energy states \cite{Dirac1930ATheoryOfElectronsAndProtons}.
The modern view is to reinterpret the creation/annihilation operators for negative energy particles as creation/annihilation operators for positive energy antiparticles.

-- Relativity suggests that if there are position operators, there should be time operators as well.

-- Canonical quantization of gravity is based on the Hamiltonian formulation of General Relativity \cite{adm2008admRepublication}, whose Hamiltonian does not have a ground state \cite{Dewitt1967QuantumTheoryOfGravityI_TheCanonicalTheory,UnruhWald1989TimeAndTheInterpretationOfCanonicalQuantumGravity}.

-- The Page-Wootters formalism for the Wheeler-DeWitt \emph{constraint equation} \cite{Dewitt1967QuantumTheoryOfGravityI_TheCanonicalTheory}, uses a clock with the Hamiltonian \eqref{eq:time_translation_Hamiltonian} \cite{PageWootters1983EvolutionWithoutEvolution}. In \cite{Stoica2022VersatilityOfTranslations} I show that the total system has the same Hamiltonian \eqref{eq:time_translation_Hamiltonian} (Example \ref{tqs:clock}).

-- Hegerfeldt showed that, if its Hamiltonian has a ground state, a free wavefunction starting in a bounded region of space would spread instantaneously in the whole space, contradicting Special Relativity \cite{Hegerfeldt1994CausalityProblemsForFermisTwoAtomSystem}.

-- Theorem \ref{thm:discrete_change} shows that the existence of a ground state violates irreversibility of change, while the Hamiltonian $\wh{H} = -i\hbar\frac{\partial\ }{\partial\tau}$ ensures it in its strongest form.

-- Examples \ref{tqs:measurement}, \ref{tqs:clock}, and \ref{tqs:dyn_sys} include quantum systems that can be very similar to what we know so far about or world, and yet their Hamiltonian is \eqref{eq:time_translation_Hamiltonian}.
\qed
\end{reply}

\begin{remark}
\label{rem:measurements}
The standard model of quantum measurements described in \cite{BuschGrabowskiLahti1995OperationalQuantumPhysics} \S II.3.4 involves ideal (projective) but also generalized (POVM) measurements.
In \cite{BuschGrabowskiLahti1995OperationalQuantumPhysics} \S VII this model is illustrated with many physically realistic realizations, including for the Stern-Gerlach experiment and various experiments with photons. 
But Example \ref{tqs:measurement} does not necessarily show that the Hamiltonian is \eqref{eq:time_translation_Hamiltonian} in these experiments, because in most cases measurements are not ideal.
What Example \ref{tqs:measurement} shows is that the Hamiltonian \eqref{eq:time_translation_Hamiltonian} is consistent with ideal quantum measurements, but not necessarily that it describes physically realistic measurements.
\end{remark}

\section{Conclusions}
\label{s:conclusions}

We have seen that the existence of a ground state is incompatible with irreversible changes.
Both these properties seem to be important and desirable.
However, in my modest analysis I did not find a decisive proof for one or the other of the two options. There are reasons to believe that the real-world Hamiltonian has a ground state, but we have seen that this would be incompatible even with weaker forms of irreversibility.
Another objection against the Hamiltonian $\widehat{H}=-i\hbar\frac{\partial\ }{\partial\tau}$ is that it is too simple to describe a complex world like ours. But we have seen that it is definitely versatile enough to faithfully describe
the quantum representation of any deterministic time-reversible dynamical system without time loops, regardless its complexity,
and also any world containing an ideal clock or a sterile massless fermion in a certain state. In addition, the Hamiltonian of the standard model of ideal quantum measurements is of this form.

I list several questions that remain to be answered:

-- Can Hamiltonian \eqref{eq:time_translation_Hamiltonian} be rejected by experiments? If it is refuted, the price is that we can no longer be certain of our own macro state, so how can we still trust our measurements that led to its refutation?

-- Can an experiment prove violations of irreversibility of change?
This seems impossible, because even if such a violation will occur, we will not notice it, since our memories will be updated accordingly, as part of the macro state.
Since we are talking about the macro state of the entire universe, there is no other system with respect to which to identify that the macro state reversed to a past macro state.

-- Is this Hamiltonian consistent with the Standard Model of particle physics and with the theory that will result someday by incorporating gravity and the explanation of dark matter and dark energy?

-- Does this result have relevant implications, or it seems important only because of an exaggerated desire to have certainties about the macro world?

\textbf{Acknowledgement} The author thanks Basil Altaie, Almut Beige, Eliahu Cohen, Ismael Paiva, Ashmeet Singh, and Michael Suleymanou, for their valuable comments and suggestions offered to a previous version of the manuscript. Nevertheless, the author bears full responsibility for the article.


\end{document}